%% file: article_v7.tex
\documentclass{article}
\usepackage{amsmath,amsfonts,amssymb,amsthm, subfig}
\usepackage[round]{natbib}
\usepackage{algpseudocode}
\usepackage{algorithm}
\usepackage{mathtools}
\usepackage{float}
\usepackage{graphicx}
\usepackage{booktabs}
\usepackage{color}
\include{newcommandsnotbook}

\include{newcommandsarticle}

\newcommand{\art}[1]{{\color{cyan}[Art: #1]}}
\renewcommand{\art}[1]{}

\newcommand{\bh}{{{\mathrm{BH}}}}

\newcommand{\simiid}{\stackrel{\mathrm{iid}}{\sim}}

\newcommand{\simdef}{S_{\mathrm{def}}}
\newcommand{\ms}{{\text{MS}}}

\newtheorem{lemma}{Lemma}
\newtheorem{theorem}{Theorem}
\newtheorem{corollary}{Corollary}
\theoremstyle{definition}

\newtheorem{definition}{Definition}

\title{Optimal mixture weights
in multiple importance sampling}
\author{Hera Y. He
\\Stanford University
 \and Art B. Owen\\
Stanford University}
\date{November 2014}

\begin{document}
\maketitle

\art{Changed your name}
\begin{abstract}
In multiple importance sampling we combine
samples from a finite list of proposal distributions.
When those proposal distributions are used
to create control variates, it is possible
\citep{owen:zhou:2000} to bound the ratio
of the resulting variance to that of the
unknown best proposal distribution in our list.
The minimax regret arises by taking a uniform mixture
of proposals, but that is conservative when there are many components.
In this paper we optimize the mixture component
sampling rates to gain further efficiency. We
show that the sampling variance of mixture
importance sampling with control variates is jointly convex
in the mixture probabilities and control variate
regression coefficients. We also give a sequential
importance sampling algorithm to estimate the optimal
mixture from the sample data.
\end{abstract}

\section{Introduction}

Importance sampling is an essential Monte Carlo method, especially when the
integrand has a singularity or when we need to work with rare events. More generally,
when the integrand is concentrated in a region of small probability
it pays to take more than the nominal number of points from that important region.
For an expectation with respect to a nominal
distribution $p$, we sample from an importance
distribution $q$ and counter the resulting bias
by multiplying our sample values by $p/q$.

It can be very difficult to choose $q$ wisely.
A common strategy is to sample from multiple distributions $q_j$,
for $j=1,\dots,J$ corresponding to different extreme scenarios.
The simplest approach is to sample from a mixture
distribution that samples $q_j$ with probability $\alpha_j$.
An improvement on that is to sample exactly $n_j = n\alpha_j$ observations
from each $q_j$.  We  call both of these methods mixture importance
sampling, with the latter being stratified mixture importance sampling.
Multiple importance sampling \citep{veac:guib:1995} is a very large
family of estimators generalizing stratified mixture importance sampling.

In bidirectional path sampling for
graphical rendering, the scenarios underlying the $q_j$
might be descriptions of paths that light follows
between source and viewing plane
\citep{lafo:will:1993,veac:guib:1994}.
Multiple importance sampling is so important for
the movie industry, that it was mentioned among
Eric Veach's contributions for his technical Oscar in 2014.
Similar strategies are used to sample
failure modes for power systems or financial positions.
In these applications, the scenarios can be all of the specific failure
mechanisms of which we are aware.

\cite{owen:zhou:2000} combine  mixture and multiple
importance sampling with certain control variates derived from the
sampling distributions.
The combination enjoys a regret bound; essentially
you cannot do worse that way than had you used only
the sample values from the single best component. The
identity of the best component might be unknown, or
it might differ for multiple integration problems
estimated on the same sample.  We do not expect a
better lower bound in general, because the single best importance
sampler might be the only one of our $J$ choices
that has a finite variance.

Letting $n=\sum_{j=1}^Jn_j$, the sampling fraction
for $q_j$ is $\alpha_j=n_j/n$.
A standard choice is $\alpha_j=1/J$,
in which all components are sampled equally.
\cite{veac:1997} recommends this allocation, but the
examples in that thesis have only moderately large $J$.
If $J$ is very
large and only one or a small number of the mixture components
is very useful, then it is inefficient to sample them equally.
In this paper we show how to optimize the choice
of the mixture weights $\alpha_j$. We show that
the variance in mixture importance sampling is a convex function of these
weights. Furthermore, when control variates are used, the variance is
jointly convex in the control variate parameters and the
mixture weights, via a quadratic-over-linear construction
\citep{boyd:vand:2004}.
Given an estimate $\hat\alpha$ of the optimal $\alpha$ from a preliminary
sample, one can sample from the $\hat\alpha$ mixture to get
a final estimate.

Section~\ref{sec:isandcv} reviews basics of importance sampling,
and control variates to set our notation. It also mentions their combination
in a single problem, which is less well known.
Section~\ref{sec:regcon} contains our main theoretical results.
We prove regret bounds for two
versions of mixture importance sampling with control variates.
In the presence of a defensive mixture component, those methods
define equivalent families of estimators.  The new version has
a variance which is a convex function of the mixing parameter.
The variance is also jointly convex in the mixing parameter
and the control variate regression parameter.
Stratified mixture importance sampling reduces to the most popular
multiple importance sampling method, known as the balance heuristic.
Incorporating control variates into it makes the variance even smaller.
Our preferred method is stratified mixture sampling with control
variates because it attains the regret bounds for both the balance
heuristic and stratified importance sampling.
Section~\ref{sec:ais} discusses sampling based ways
to choose mixture probabilities using the convexity results from
Section~\ref{sec:regcon}.
Section~\ref{sec:examples} has two numerical examples,
one from  a singular integrand, and one from a rare event.
Section~\ref{sec:conclusions} has our conclusions and
some discussion of the literature.

\section{Importance sampling and control variates}\label{sec:isandcv}
The problem we consider is approximating
$$
\mu = \int f(\bsx)p(\bsx)\rd\bsx
$$
for an integrand $f$ and probability density function $p$.
We use notation for a continuously distributed random
variable $\bsx$ on Euclidean space but the results
generalize to other settings, such as discrete random variables,
with only minor changes.

The standard Monte Carlo method for estimating $\mu$
takes a sample $\bsx_i\simiid p$ and produces the
estimate $\hat\mu=(1/n)\sum_{i=1}^nf(\bsx_i)$.
When $\sigma^2 = \int (f(\bsx)-\mu)^2\rd\bsx<\infty$
then $\var(\hat\mu) = \sigma^2/n$.

For $f$ corresponding to rare events or singular integrands,
among other problems, there is region of small probability
under $p$ that is very important to sample from.
Let $q$ be any probability density function satisfying
$q(\bsx)>0$ whenever $f(\bsx)p(\bsx)\ne 0$. Then
for $\bsx_1,\dots,\bsx_n$ sampled from $q$,
the importance sampling estimate
$$
\hat\mu_q =
\frac1n\sum_{i=1}^n\frac{f(\bsx_i)p(\bsx_i)}{q(\bsx_i)}
$$
satisfies $\e(\hat\mu_q)=\mu$.
Let $D(q) = \{\bsx\mid q(\bsx)>0\}$.
The variance of $\hat\mu_q$ is $\sigma^2_q/n$
where
$$\sigma^2_q =
\frac1n\biggl( \int_{D(q)}\frac{f(\bsx)^2p(\bsx)^2}{q(\bsx)}\rd\bsx-\mu^2
\biggr)
=
\int_{D(q)}
\frac{  (f(\bsx)p(\bsx)-\mu q(\bsx))^2}{q(\bsx)}\rd\bsx.$$
Sampling from the nominal distribution $p$
has variance $\sigma^2_p/n$
where $\sigma^2_p = \int (f(\bsx)-\mu)^2p(\bsx)\rd\bsx=\sigma^2$.
When $f\ge0$ the optimal $q$ is proportional to $fp$ and
it yields variance $0$.
Good choices for $q$ sometimes yield $\sigma^2_q\ll\sigma^2_p$.
Poor choices can have $\sigma^2_q=\infty$ even when
$q$ is nearly proportional to $fp$,
because $q(\bsx)$ appearing in the denominator
of the variance integral may be small.
It is often convenient to work with the following mean square,
$$
\ms(\hat\mu_q) \equiv \mu^2+\sigma_q^2 = \int_{D(q)}
\frac{  (f(\bsx)p(\bsx))^2}{q(\bsx)}\rd\bsx,
$$
instead of with the variance.

Now suppose that $h$ is a vector-valued function for which we know $\e_p(h(x))$.
That is, $\int h(\bsx)p(\bsx)\rd\bsx =\theta\in\real^K$, where
$\theta$ is known.
For any $\beta\in\real^K$
$$
\hat\mu_\beta = \frac1n\sum_{i=1}^n f(\bsx_i)-\beta^\tran h(\bsx_i) + \beta^\tran\theta
$$
is an unbiased estimate of $\mu$.
Then
$$
 \var(\hat\mu_\beta) = \frac{\sigma^2_\beta}{n},\quad
\text{where}\quad
\sigma^2_\beta = \int (f(\bsx)-\mu-(h-\theta)^\tran\beta)^2\rd\bsx.
$$
The (unknown) variance-optimal $\beta$ is
$\beta_* = \var(h(\bsx))^{-1}\cov(h(\bsx),f(\bsx))$.
A standard practice is to estimate $\beta_*$ by
its least squares estimate $\hat\beta$ and
then use $\hat\mu_{\hat\beta}$. This estimator has
an asymptotically negligible bias of order $O(1/n)$
as $n\to\infty$ with $J$ fixed.
While the bias is ordinarily ignored,
it can be eliminated entirely by using different
observations to compute $\hat\beta$ than those
used to compute $\hat\mu$ or by cross-validatory
schemes \citep{avra:wils:1993}, at the expense of
increased bookkeeping.

A sample least squares estimator
is typically root-$n$ consistent: $\hat\beta = \beta_* + O_p(n^{-1/2})$.
An error of size $n^{-1/2}$ in $\beta$
raises the variance of the estimate for
$\mu$ by a multiplicative factor of
$1+O(n^{-1})$ because that variance
is locally quadratic around $\beta_*$.
Because the effect of estimated $\hat\beta$
is negligible, we analyze the effect
of control variates as if we were
using the unknown optimal coefficient $\beta_*$.


\cite{owen:zhou:2000} incorporate control variates
into importance sampling.
Our approach differs slightly from that paper
which defines control variates via known integrals $\int h(\bsx)\rd\bsx$,
rather than known expectations.
Their $h(\bsx)$ is our $h(\bsx)p(\bsx)$.

We may incorporate a control variate
into importance sampling via either
\begin{align*}
\tilde \mu_{q,\beta} & = \frac1n\sum_{i=1}^n
\frac{ \bigl[f(\bsx_i)-\beta^\tran h(\bsx_i)\bigr]p(\bsx_i)}{q(\bsx_i)} +\beta^\tran\theta,
\quad\text{or}\\
\hat \mu_{q,\beta} & = \frac1n\sum_{i=1}^n
\frac{ \bigl[f(\bsx_i)-\beta^\tran (h(\bsx_i)-\theta)\bigr]
p(\bsx_i)}{q(\bsx_i)}.
\end{align*}
Both estimators $\tilde\mu_{q,\beta}$ and $\hat\mu_{q,\beta}$
require that we can compute the ratio $p(\bsx)/q(\bsx)$.

In ordinary control variates, where $q=p$, these two estimates coincide.
They also coincide for $\theta=0$, or even for
$\beta^\tran\theta=0$,  but in general they are different:
$$
\hat \mu_{q,\beta} - \tilde \mu_{q,\beta}  =
\frac{\beta^\tran\theta}n
\sum_{i=1}^n\Bigl(\frac{p(\bsx_i)}{q(\bsx_i)}-1\Bigr).
$$
If $q(\bsx)>0$ whenever $p(\bsx)>0$,
then $\e( \hat\mu_{q,\beta})=\e( \tilde\mu_{q,\beta})=\mu$ for all $\beta\in\real^K$.
Strictly speaking, we only need $q(\bsx)>0$ whenever $h(\bsx)p(\bsx)\ne0$ or
$(h(\bsx)-\theta)p(\bsx)\ne0$ for unbiasedness of $\tilde\mu_{q,\beta}$
and $\hat\mu_{q,\beta}$ respectively.

We find that
\begin{align}
n\var(\tilde\mu_{q,\beta}) & =
\int \Bigl(\frac{fp}q - \mu - \beta^\tran \Bigl(
\frac{hp}q - \theta
\Bigr)
\Bigr)^2q\rd \bsx,\quad\text{and}\label{eq:variscvtilde}\\
n\var(\hat\mu_{q,\beta}) & =
\int \Bigl(\frac{fp}q - \mu - \beta^\tran \frac{ (h-\theta)p}q
\Bigr)^2q\rd \bsx.\label{eq:variscvhat}
\end{align}
Here and below we omit the argument $\bsx$
to shorten some expressions.
Neither estimator universally dominates the other.

The optimal $\beta$ is generally unknown. For $\wt\mu_{q,\beta}$
we see from~\eqref{eq:variscvtilde} that the optimal $\beta$ is
the slope in a $q$-weighted least squares regression of $fp/q$
on $hp/q-\theta$. For $\hat\mu_{q,\beta}$,
and using~\eqref{eq:variscvhat}, we find that the optimal $\beta$
is from a $q$-weighted regression of $fp/q$ on a slightly different predictor,
$hp/q -\theta p/q$.

Given $n$ observations $\bsx_i$ from $q$, we can estimate $\beta$
by running the corresponding regressions. Unweighted least squares
is then appropriate  because the sample values are already
sampled from $q$.

We will find the mean square formulation useful.
The variance of $\tilde\mu_{q,\beta}$ satisfies
$n\var(\tilde\mu_{q,\beta}) = \ms(\tilde\mu_{q,\beta}) -\mu^2$ where
\begin{align}
 \ms(\tilde\mu_{q,\beta})&=
 \int \frac{ [ (f -\beta^\tran h)p+\beta^\tran \theta q ]^2 }{q}\rd\bsx,\quad\text{and}\label{eq:msqcvtilde}\\
 \ms(\hat\mu_{q,\beta})&=
 \int \frac{\bigl[f -\beta^\tran(h-\theta) \bigr]^2p^2  }{q}\rd\bsx\label{eq:msqcvhat}
\end{align}
for $n\var(\hat\mu_{q,\beta})=\ms(\hat\mu_{q,\beta})-\mu^2$.

Though neither estimator $\hat\mu_{q,\beta}$ or
$\tilde\mu_{q,\beta}$ will always have lower variance,
they do differ on some other meaningful properties.
We find in Section~\ref{sec:regcon} that when importance sampling from mixtures,
$\tilde\mu_{q,\beta}$ offers a regret bound compared
\art{Took out all the extra defensive mixture statements.}
to  the individual mixture components.
The estimator $\hat\mu_{q,\beta}$ also offers a regret
bound 
and moreover, it can be optimized jointly
over $\alpha$ and $\beta$ by convex optimization.
That makes it possible to develop an adaptive importance
sampling strategy.

\section{Regret bounds and convexity}\label{sec:regcon}

In this section we produce regret bounds describing
how much variance might be increased when sampling
from multiple distributions compared to using the unknown optimal
sampling distribution. We begin with sampling from
mixtures. Sampling a deterministic number of observations
from each mixture component is then at least as good.
Then we look at Veach's multiple importance sampling.
Finally we prove convexity results.

\subsection{Mixture importance sampling}

Here we sample
from a mixture of distributions $q_1,q_2,\dots,q_J$ of the form
$q_\alpha = \sum_{j=1}^J\alpha_jq_j$
where $\alpha_j\ge0$ and $\sum_{j=1}^J\alpha_j=1$.
We write
$$S=\{ (\alpha_1,\dots,\alpha_J)
\mid \alpha_j \ge 0, \sum_{j=1}^J\alpha_j=1\}$$
for the simplex of mixture weights $\alpha$
and $S^0$ for the interior of $S$ in which
every $\alpha_j>0$.
We also write $\hat\mu_\alpha$ as a shorthand for
$\hat\mu_{q_\alpha}$ when the list of $q_j$ is fixed.
We choose $q_\alpha$ so that $q_\alpha(\bsx)>0$
whenever $f(\bsx)p(\bsx)\ne 0$.
The variance of $\hat\mu_\alpha$ is $\sigma^2_\alpha/n$ where
\begin{align}\label{eq:sigsqalpha}
\sigma^2_\alpha
=\int_{D(q_\alpha)} \frac{\bigl( f(\bsx)p(\bsx)
-\mu\sum_{j=1}^J\alpha_jq_j(\bsx)\bigr)^2}{\sum_{j=1}^J\alpha_jq_j(\bsx) }\rd \bsx.
\end{align}

It is a good practice to include a defensive mixture
component, $q_j=p$.  Without loss of generality
we let $q_1=p$ in defensive mixtures.
The simplex of defensive mixture weights is
$$\simdef = \{\alpha\in S\mid \alpha_1>0\}.$$
For a defensive mixture, $p/q_\alpha =p/(\alpha_1p+\cdots+\alpha_Jq_J)\le1/\alpha_1$.
This bounded density ratio leads
to the variance bound
\begin{align}\label{eq:defbound}
\sigma^2_\alpha \le \frac1{\alpha_1}\bigl( \sigma^2_p + (1-\alpha_1)\mu^2\bigr),
\end{align}
which follows from the $J=2$ case in
\cite{hest:1988}.

While defensive importance sampling yields a variance
not so much more than $\sigma^2_p$, it might
yield a variance much larger than $\sigma^2_{q_j}$
for some $j>1$.
\cite{owen:zhou:2000} counter this problem via control variates
devised from the mixture components.

The components $q_j$  integrate to $1$ because they
are probability densities.
For $j=1,\dots,J$, let
\begin{align}\label{eq:defineh}
h_j(\bsx) = \begin{cases} q_j(\bsx)/p(\bsx), & p(\bsx)>0\\
0, & \text{else.} \end{cases}
\end{align}
If $p(\bsx)>0$ whenever $q_j(\bsx)>0$,
then $\theta_j = \int h_j(\bsx)p(\bsx)\rd\bsx = 1$.
In other settings we may still know $\theta_j$, though
some special purpose argument would be required to find it.

\begin{definition}\label{def:supportmix}
Given a density $p$,
the densities $q_1,\dots,q_J$ and the point $\alpha\in S$
together satisfy the \emph{support conditions}
if $p(\bsx)>0$ whenever any $q_j(\bsx)>0$ and
$\sum_{j=1}^J\alpha_jq_j(\bsx)>0$ whenever $p(\bsx)>0$.
\end{definition}

Let $h(\bsx) = (h_1(\bsx),\dots,h_J(\bsx))^\tran$
and $\theta = (\theta_1,\dots,\theta_J)^\tran$.
The combined control variate with importance sampling
estimates are now
\begin{align}
\tilde\mu_{\alpha,\beta} &= \frac1n\sum_{i=1}^n
\frac{ [f(\bsx_i)-\beta^\tran h(\bsx_i)]p(\bsx_i)}{q_\alpha(\bsx_i)} +\beta^\tran\theta,
\quad\text{and}\label{eq:tildemuab}
\\
\hat\mu_{\alpha,\beta} &=
\frac1n\sum_{i=1}^n
\frac{ [f(\bsx_i)-\beta^\tran (h(\bsx_i)-\theta)]p(\bsx_i)}{q_\alpha(\bsx_i)}.\label{eq:hatmuab}
\end{align}
If $\alpha\in S^0$ and  $q_\alpha(\bsx)>0$ whenever $f(\bsx)p(\bsx)\ne 0$
then $\e(\tilde\mu_{\alpha,\beta})=\e(\hat\mu_{\alpha,\beta})=\mu$
for $\bsx_i\sim q_\alpha$.
Note that if we have additional control variates
$h_j$ for $J<j\le K$ we can incorporate them
directly into $h(\bsx)$ and $\theta$.

\cite{owen:zhou:2000} noticed that the regression to
estimate $\beta$ is always rank deficient for control
variates defined via the sampling densities.
That  problem can be countered by either dropping one of the
control variates or by computing the least
squares estimate $\hat\beta$ via the singular value decomposition (SVD).
In the present formulation of
the problem, the regressor $(h_1-1)p/q_\alpha$
in $\hat\mu_{\alpha,\beta}$
(equation~\eqref{eq:hatmuab}) satisfies
$$(h_1(\bsx)-1)p(\bsx)/q_\alpha(\bsx) = (p(\bsx)-p(\bsx))/q_\alpha(\bsx)=0$$
for all $\bsx$. The defensive
control variate is by definition zero and we
simply drop it from the regression.
The regressors $h_jp/q_\alpha$
in estimator $\tilde\mu_{\alpha,\beta}$ (equation~\eqref{eq:tildemuab})
satisfy $\sum_{j=1}^J\alpha_j h_jp/q_\alpha
=\sum_{j=1}^J\alpha_j q_j/q_\alpha =1,$ at all $\bsx$.
One might just drop the control variate with smallest
sample variance, or use the SVD. 
\art{fixed the summation typo}

Theorem~\ref{thm:smallregret} below shows
there always exists a $\beta$ that makes
the estimator $\wt\mu_{\alpha,\beta}$  nearly
as good as ordinary importance sampling
from the best $q_j$ in a list.

\begin{theorem}\label{thm:smallregret}
Let $\mu = \int f(\bsx)p(\bsx)\rd\bsx$
for a probability density $p$.
Let densities $q_1,\dots,q_J$ and $\alpha\in S$
satisfy the support conditions of Definition~\ref{def:supportmix}.
Let $\tilde\mu_{\alpha,\beta}$ be the estimator
given by~\eqref{eq:tildemuab} for independent
$\bsx_i\sim q_\alpha = \sum_{j=1}^J\alpha_jq_j$
with $h_j$ defined by~\eqref{eq:defineh}.
If $\beta_*$ is any minimizer of
$\var( \tilde\mu_{\alpha,\beta})$, then
$$
\var( \tilde\mu_{\alpha,\beta_*})
\le \frac1n \min_{1\le j\le J} {\sigma^2_{q_j}}/{\alpha_j}.
$$
\end{theorem}
\begin{proof}
See Section~\ref{sec:thm:smallregret}.
\end{proof}

The quantity $\sigma^2_{q_k}/(n\alpha_k)$ is the variance
we would have gotten from importance sampling
using only $n\alpha_k$ observations all taken from $q_k$.
The other variances $\sigma^2_{q_j}$ could be arbitrarily
large and yet they cannot make the estimate $\wt\mu_{\alpha,\beta_*}$
worse than importance sampling with $n\alpha_k$ observations
from $q_k$.
We can improve on this result as follows.

\begin{corollary}\label{cor:smallregretbeta}
Under the conditions of Theorem~\ref{thm:smallregret}
$$
\min_\beta\var( \tilde\mu_{\alpha,\beta})
\le \min_{1\le j\le J} \min_{\beta_j}\frac{\sigma^2_{q_j,\beta_j}}{n\alpha_j}.
$$
\end{corollary}
\begin{proof}
See Section~\ref{sec:cor:smallregretbeta}.
\end{proof}

Next we show a regret bound for $\hat\mu_{\alpha,\beta}$.

\begin{theorem}\label{thm:smallregrethat}
Let $\mu = \int f(\bsx)p(\bsx)\rd\bsx$
for a probability density $p$.
Let densities $q_1,\dots,q_J$ and $\alpha\in S$ satisfy the support conditions
of Definition~\ref{def:supportmix}.
Let $\hat\mu_{\alpha,\beta}$ be the estimator
given by~\eqref{eq:hatmuab} for independent
$\bsx_i\sim q_\alpha = \sum_{j=1}^J\alpha_jq_j$
with $h_j$ defined by~\eqref{eq:defineh}.
If $\beta_*$ is any minimizer of
$\var( \hat\mu_{\alpha,\beta})$, then
$$
\var( \hat\mu_{\alpha,\beta_*})
\le \frac1n \min_{1\le j\le J} {\sigma^2_{q_j}}/{\alpha_j}.
$$
\end{theorem}
\begin{proof}
See Section~\ref{sec:thm:smallregrethat}.
\end{proof}

\begin{corollary}\label{cor:smallregrethatbeta}
Under the conditions of Theorem~\ref{thm:smallregrethat}
$$
\min_\beta\var( \hat\mu_{\alpha,\beta})
\le \min_{1\le j\le J} \min_{\beta_j}\frac{\sigma^2_{q_j,\beta_j}}{n\alpha_j}.
$$
\end{corollary}
\begin{proof}
See Section~\ref{sec:cor:smallregretbeta}.
\end{proof}

In practice it is wise to stratify
the sampling \citep{hest:1995}.
In stratification we take $n_j$ observations
from $q_j$ instead of sampling iid from $q_\alpha$ with $\alpha_j=n_j/n$.
Stratification of mixture importance sampling reduces the variance
of $\hat\mu_\alpha$ by
\begin{align*}
\frac1n\sum_{j=1}^J\alpha_j(\mu_j-\mu)^2,\quad\text{where}\quad
\mu_j = \int \frac{f(\bsx)p(\bsx)}{q_\alpha(\bsx)}q_j(\bsx)\rd\bsx.
\end{align*}
 Stratification can reduce but cannot increase
the variance of $\hat\mu_\alpha$,
$\hat\mu_{\alpha,\beta}$, $\tilde\mu_\alpha$
and $\tilde\mu_{\alpha,\beta}$.  Thus when
each $n\alpha_j$ is a positive integer,
Theorems~\ref{thm:smallregret}
and~\ref{thm:smallregrethat}
as well as Corollaries~\ref{cor:smallregretbeta}
and~\ref{cor:smallregrethatbeta} all apply
to stratified mixture sampling.
We will use $\var_\strat$ to denote
the variance under stratified sampling of mixture
sampling estimators.

\subsection{Multiple importance sampling}

Multiple importance sampling \citep[Chapter 9]{veac:1997}
takes $n_j$ observations
from $q_j$ for $j=1,\dots,J$ for a total
of $n=\sum_{j=1}^Jn_j$ observations.
The multiple importance sampling estimator
is defined in terms of a partition of unity: $J$ functions
that add up to $1$.  They must obey some
additional conditions given by this definition:
\begin{definition} \label{def:partunity}
A multiple sampling partition of unity for distributions
$p,q_1,\dots,q_J$ is a set of real-valued functions
$\omega_j(\bsx)$ for which $\sum_{j=1}^J\omega_j(\bsx)=1$
at all $\bsx$ with $p(\bsx)>0$, and for which $\omega_j(\bsx)=0$
whenever $q_j(\bsx)=0$.
\end{definition}

The multiple importance sampling estimator is
\begin{align}\label{eq:mulis}
\check\mu_\omega =
\sum_{j=1}^J\frac1{n_j}
\sum_{i=1}^{n_j} \omega_j(\bsx_{ij})\frac{f(\bsx_{ij})p(\bsx_{ij})}{q_j(\bsx_{ij})}.
\end{align}
We easily find that $\e(\check\mu_\omega)=\mu$ when
$\bsx_{ij}\sim q_j$ independently
and $\omega_j$ satisfy Definition~\ref{def:partunity}.
Multiple importance sampling reduces to
classic stratified sampling when the $\omega_j$ have
disjoint supports.  In its general form, it is like a stratification
that allows overlapping strata.

Among the many choices for $\omega_j$,
\cite{veac:1997} favors the balance heuristic
$$\omega_j^\bh(\bsx)
= \frac{n_jq_j(\bsx)}{\sum_{k=1}^Jn_kq_k(\bsx)}
= \frac{n_jq_j(\bsx)}{nq_\alpha(\bsx)},
\quad\text{for}\quad \alpha_j = n_j/n.
$$
With this choice
\begin{align}\label{eq:bhisstratmixis}
\check\mu_{\omega^\bh} =
\sum_{j=1}^J\frac1{n_j}
\sum_{i=1}^{n_j}
\frac{n_jq_j(\bsx_{ij})}{nq_\alpha(\bsx_{ij})}
\frac{f(\bsx_{ij})p(\bsx_{ij})}{q_j(\bsx_{ij})}
=
\frac1n\sum_{j=1}^J\sum_{i=1}^{n_j}\frac{f(\bsx_{ij})p(\bsx_{ij})}{q_\alpha(\bsx_{ij})}.\end{align}
Notice that the coefficient of $f(\bsx_{ij})$ does not depend on which mixture
component $\bsx_{ij}$ was sampled from.
Equation~\eqref{eq:bhisstratmixis} shows that the balance heuristic
is the same as mixture sampling with stratification,
but without control variates. Therefore control variates with
mixture sampling improves on the balance heuristic.


\citet[Theorem 9.2]{veac:1997}
proves that the balance heuristic satisfies
\begin{align}\label{eq:veachbhbound}
\var(\check\mu_{\omega^\bh}) \le \var(\check\mu_\omega) +
\Bigl( \frac1{\min_jn_j}-\frac1n\Bigr)\mu^2
\end{align}
for any $\omega$ satisfying Definition~\ref{def:partunity}.
Equation~\eqref{eq:veachbhbound} is more general
than the Theorems of the previous section because
it covers any alternative partition of unity, not just
sampling from the $q_j$ themselves.
The bound~\eqref{eq:veachbhbound} also applies directly
to $\hat\mu_{\alpha}$ under stratified mixture sampling.
For the alternative heuristic that simply takes $\omega_j(\bsx)=1$
(and all other $\omega_k=0$),
equation~\eqref{eq:veachbhbound} simplifies to
$$
\var_\strat(\hat\mu_\alpha)
= \var(\check\mu_{\omega^\bh})
\le \frac{\sigma^2_j  }{n_j}
+ \Bigl( \frac1{\min_jn_j}-\frac1n\Bigr)\mu^2,
$$
while from Theorems~\ref{thm:smallregret}
and~\ref{thm:smallregrethat} we have
\begin{align}\label{eq:nomusq}
\min_\beta\var_\strat(\tilde\mu_{\alpha,\beta})
\le\min_{1\le j\le J} \frac{\sigma^2_j  }{n_j},
\quad\text{and}\quad
\min_\beta\var_\strat(\hat\mu_{\alpha,\beta})
\le \min_{1\le j\le J}\frac{\sigma^2_j  }{n_j},
\end{align}
for $\alpha\in S^0$.
The control variates remove
the multiple of~$\mu^2$ which can be arbitrarily
large if one adds a large enough constant to $f$.

\cite{veac:1997} remarks that the
balance heuristic is not competitive
when one of the component samplers is especially
effective.  For such low variance problems
he proposes alternative heuristics, such as
a power heuristic that takes $\omega_j\propto q_j^r$
with $r=2$ singled out as a useful case.
Equation~\eqref{eq:nomusq} shows
that incorporating control variates automatically
yields competitive estimates for all cases
including those high accuracy ones.

If $J$ is not too large then a reasonable
strategy is to take each $\alpha_j=1/J$.
The variance under either multiple or mixture importance sampling will be at most $J$ times the
best single component variance and possibly much better.
When the individual $\sigma^2_j$ vary over many orders of magnitude,
losing a factor of $J$ compared to the best is acceptable.
However when $J$ is large we may want to optimize
over $\alpha$.

\subsection{Convexity}
Here we provide convexity results for the
estimation variances. These results support adaptive
minimization strategies where one alternates between
minimizing an estimated variance and sampling
from the best estimate of the optimal~$\alpha$.
Theorem~\ref{thm:convexalpha} covers the
case of mixture importance sampling without
control variates.

\begin{theorem}\label{thm:convexalpha}
The quantity $\sigma^2_\alpha=n\var(\hat\mu_\alpha)$ in equation~\eqref{eq:sigsqalpha} is
a convex function of $\alpha\in S^0$.
\end{theorem}
\begin{proof}
See Section~\ref{sec:thm:convexalpha}
\end{proof}

To handle control variates, the optimal $\alpha$
depends on the regression vector $\beta$ and
vice versa. Theorem~\ref{thm:jointconvex} shows
that the variance of $\hat\mu_{\alpha,\beta}$ is jointly
convex in these two vectors.

\begin{theorem}\label{thm:jointconvex}
The quantity $\sigma^2_{\alpha,\beta} = n\var(\hat{\mu}_{\alpha, \beta})$ in 
equation~\eqref{eq:hatmuab} is
a jointly convex function of $\alpha\in S^0$ and $\beta\in\real^{J}$.
\end{theorem}
\begin{proof}
See Section~\ref{sec:thm:jointconvex}.
\end{proof}

We can apply convex optimization to choose $\alpha$
for $\hat\mu_\alpha$ or $\hat\mu_{\alpha,\beta} $ but
not directly for $\wt\mu_{\alpha,\beta}$.
Lemma~\ref{lem:samespace} shows that
when a defensive component is present then the set 
of estimators attained by the two estimator forms,
$\tilde\mu_{\alpha,\beta}$ and $\hat\mu_{\alpha,\beta}$ coincide.
Therefore if we choose $\alpha$ to minimize the variance of $\hat\mu_{\alpha,\beta}$ 
and the optimal $\alpha_*\in \simdef$,
then $\alpha_*$ is optimal for the variance of $\tilde\mu_{\alpha,\beta}$ too.
\begin{lemma}\label{lem:samespace}
Suppose that $\alpha\in \simdef$ and choose $\beta\in\real^J$. 
Let $\alpha$ and the densities $q_1,\dots,q_J$ satisfy 
the support condition of Definition~\ref{def:supportmix}
and let $h_j$ be defined by~\eqref{eq:defineh}. 
Then there exists a $\gamma\in\real^J$ such that 
$\hat\mu_{\alpha,\gamma} =\tilde\mu_{\alpha,\beta}.$
\end{lemma}
\begin{proof}
See Section~\ref{sec:lem:samespace}
\end{proof}

We often have to estimate multiple integrals
from the same sample values $\bsx_{ij}$.
For integrands $f_k$, $k=1,\dots,K$,
let $\hat\mu_{\alpha,\beta_k}^{(k)}$ be the
estimate~\eqref{eq:hatmuab} with $f=f_k$,
$\alpha\in S$ and $\beta_k\in\real^J$.
For weights $w_k>0$,
we might make the tradeoff via either
$$
\max_{1\le k\le K}w_k\var(\hat\mu^{(k)}_{\alpha,\beta_k})
\quad\text{or}\quad
\sum_{k=1}^Kw_k\var(\hat\mu^{(k)}_{\alpha,\beta_k}).
$$
Both of these criteria inherit convexity from
Theorem~\ref{thm:jointconvex}.  Choosing not to use control
variates corresponds to constraining $\beta_k=0$ and the
multi-integral criteria remain convex in $\alpha$.

\section{Choosing $\alpha$}\label{sec:ais}

Given a sample $\bsx_i\simiid q_{\alpha'}$
for $\alpha'\in S^0$,
we may use it to estimate an improved
mixture vector $\alpha$.
From the proof of Theorem~\ref{thm:jointconvex}
we know that the sampling variance using $\alpha\in S^0$ satisfies
\begin{align}\label{eq:varsigsq2}
\sigma^2_{\alpha,\beta} + \mu^2
=
\int_D
\frac{\bigl( f(\bsx)p(\bsx) + \sum_{J=1}^J\beta_j(q_j(\bsx)-p(\bsx))
\bigr)^2}{q_{\alpha}(\bsx)}
\rd\bsx.
\end{align}
An unbiased estimate of the right hand side of~\eqref{eq:varsigsq2} is
\begin{align}\label{eq:criterion}
\frac1n\sum_{i=1}^n
\frac{\bigl( f(\bsx_i)p(\bsx_i) + \sum_{J=1}^J\beta_j(q_j(\bsx_i)-p(\bsx_i))
\bigr)^2}{q_{\alpha}(\bsx_i)q_{\alpha'}(\bsx_i)}.
\end{align}
The quantity in~\eqref{eq:criterion} is jointly
convex in $\alpha\in S^0$ and $\beta\in\real^J$.
Therefore we can use convex optimization to
find the minimizers $(\alpha_*,\beta_*)$
of~\eqref{eq:criterion} and then generate a new sample
using $\alpha_*$.

Equation~\eqref{eq:criterion} may be used to devise
adaptive importance sampling methods in which early
function evaluations are used to change the value of $\alpha$
used to generate later ones.  A full exploration of adaptive
strategies is beyond the scope of this paper. Our numerical
examples in Section~\ref{sec:examples} use a two stage method in which a pilot sample
is used to estimate $\alpha$ for a final sample that is used
for the estimate.

\subsection{Bounding $\alpha_j$ away from zero}

It is risky to sample with a value of $\alpha_j$
that is so small that we may not get a reasonable
number $n_j$ of values from $q_j$.
We may prefer to impose constraints $\alpha_j>\epsilon$
on each sampling fraction.
For $\epsilon>0$ define
\[
S^{\epsilon} = \Bigl\{(\alpha_1, \cdots, \alpha_J) \mid \alpha_j \geq
\epsilon, \sum\limits_{j = 1}^J \alpha_j = 1\Bigr\}.
\]
Theorem~\ref{thm:safetyepsilon} shows that we pay
a remarkably small variance penalty
when requiring $\alpha\in S^\epsilon$.

\begin{theorem}\label{thm:safetyepsilon}
Let $\alpha$ and $\beta$ minimize
$\var(\tilde\mu_{\alpha,\beta})$ over
$\alpha\in S$ and $\beta\in\real^J$.
For $0<\epsilon<1/J$
let $\delta$ and $\gamma$
minimize $\var(\tilde\mu_{\delta,\gamma})$
over $\delta\in S^\epsilon$ and $\gamma\in\real^J$.
Then
\[
\var(\tilde\mu_{\delta,\gamma})
\leq (1 +\eta(\alpha))
\var(\tilde\mu_{\alpha, \beta}),\]
where
$$\eta(\alpha) = \frac{1-\epsilon(J-K)}{1-\epsilon J},
\quad\text{for}\quad K=K(\alpha) = \#\{\alpha_j<\epsilon\}.$$
The same holds for the $\hat\mu_{\alpha,\beta}$
and $\hat\mu_{\delta,\gamma}$ analogously defined.
\end{theorem}
\begin{proof}
See Section~\ref{sec:thm:safetyepsilon}.
\end{proof}

The penalty factor $1+\eta(\epsilon)$ in Theorem~\ref{thm:safetyepsilon} is asymptotically
$1+K\epsilon+O(\epsilon^2)$ as $\epsilon\to 0$, and
$1+K\epsilon \le 1+(J-1)\epsilon$.
The price to be paid is bounded by a term nearly linear in the number
of components with $\alpha_j<\epsilon$.
If for example we insist on at least $k$ (e.g. $k=10$)
observations from each $q_j$, then $\epsilon = k/n$
and $1+\eta(\epsilon) \le (1-\epsilon)/(1-J\epsilon)
= (1-k/n)/(1-Jk/n)$.

For any given $\alpha$, a suboptimal $\hat\beta$
generally increases variance by a factor
of $1+O(\Vert\hat\beta-\beta_*\Vert^2)$ because
the variance function is twice continuously differentiable
on $\real^J$.
A suboptimal $\alpha$ can multiply the variance by
$1+O(\Vert\hat\alpha-\alpha_*\Vert)$.
The reason for the different rate is that the optimum
$\alpha_*$ can be on the boundary of the simplex $S$
and there is no reason to suppose that the gradient
of the variance with respect to $\alpha$ vanishes there.

\subsection{Optimization}

We found that off-the-shelf convex optimization codes we tried were
slow. As a result it was worthwhile to write our own code for
the specific case we need in this problem. We use standard
convex optimization methods, but make modifications to
improve robustness.

Our workhorse numerical optimization problem is
\begin{equation}\begin{split}
\text{minimize}\quad &f_0(\alpha,\beta)= \sum_{i=1}^n \frac{ (Y_i - X_i^\tran\beta)^2}{Z_i^\tran\alpha}\\
\text{subject to}\quad & \beta\in\real^K,\\
& \alpha_j > \epsilon_j\ge0, \ j=1,\dots,J, \ \text{and}\\
& \sum_{j=1}^J\alpha_j < 1+\eta,\ \text{for $\eta\ge0$},
\end{split}
\end{equation}
which is equivalent to~\eqref{eq:criterion}.  In this notation
$Y\in\real^n$ represents $fp$, $X\in\real^{n\times K}$
represents $q_j-p$ for $j>1$
(and possibly some additional control variates) and
$Z\in[0,\infty)^{n\times J}$ represents $q_j(\bsx)$.
We have multiplied
the objective function by $n$.
Valid problem data have $X$ of full rank $K\le n$
and no row of $Z$ equals $(0,\dots,0)$.

The constraints are defined by parameters $\epsilon_j\ge0$
and $\eta\ge0$.
We assume that there is at least one feasible vector $\alpha$.
We do not need to impose the constraint $\sum_{j=1}^J\alpha_j=1$.
The way the problem is formulated makes that constraint equivalent
to $\sum_{j=1}^J\alpha_j \le1$. As a result we do not have to mix
equality and inequality constraints.  We use a very small $\eta$,
usually $\eta=0$, and then renormalize the solution $\alpha$ to sum to one.
We can choose a different $\epsilon_j$ for each $\alpha_j$. For
example we might want a higher lower bound on the defensive
component than on the others.

We use a barrier method,
Algorithm 11.1 of \cite{boyd:vand:2004}, with function
\begin{align}\label{eq:frho}
f(\alpha,\beta;\rho)
= f(\alpha,\beta) - \rho\sum_{j=1}^J\log(\alpha_j-\epsilon_j)
-\rho\log\biggl(1+\eta -\sum_{j=1}^J\alpha_j\biggr),
\end{align}
for $\rho>0$.   For each $\rho$,  this function is jointly convex
over the convex set of allowable $(\alpha,\beta)$.

We minimize~\eqref{eq:frho} by a damped Newton iteration,
Algorithm 9.5 of \cite{boyd:vand:2004}
using a backtracking line search  (their Algorithm 9.2).
The solutions are denoted $\alpha(\rho)$ and $\beta(\rho)$.
Beginning with $\rho=\rho_1$ we decrease $\rho$
by a factor of $\phi=2$, taking $\rho_\ell = \rho_{\ell-1}/\phi$
for $\ell \ge2$,  minimizing $f(\alpha,\beta;\rho_\ell)$
at step $\ell\ge1$ obtaining $\alpha^{(\ell)}=\alpha(\rho_{\ell})$
and $\beta^{(\ell)}=\beta(\rho_{\ell})$,
using $\alpha^{(\ell-1)}$ and $\beta^{(\ell-1)}$  \art{fixed typo}
as starting values.  To start the whole algorithm we use
$\beta^{(0)}=0$ and $\alpha^{(0)}$ given by
$$
\alpha^{(0)}_{j} = \epsilon_j+\delta,\quad\text{where}\quad
\delta = \frac{1+\eta-\sum_{j=1}^J\epsilon_j}{J+1}.$$
\art{
If we equate the quantities inside the $\log(\cdot)$ penalties,
we find 
$$\delta = 1+\eta -\sum \alpha_i=1+\eta -\sum \epsilon_j -J\delta$$
which is how I got $\delta$ above.  Are you equating different things?
}
The constraints are feasible if and only if $\delta>0$.
This choice of starting value equalizes all of the $J+1$ logarithmic
penalty terms.
We continue decreasing $\rho$ until we reach $\rho_\ell$
where $(J+1)\rho_\ell <f_0(\alpha^{(\ell)},\beta^{(\ell)})\varepsilon$,
for a tolerance $\varepsilon>0$.
At this point we have a certificate that
$$
f_0(\alpha^{(\ell)},\beta^{(\ell)})
\le
(1+\varepsilon)
\min_{\alpha,\beta}f_0(\alpha,\beta).$$

We have found it helpful to consider preconditioning the
Hessian matrix for our Newton steps.
The intuition is that if $\alpha$ and $\beta$ are on very different
scales, scaling them to be comparable will reduce the
conditon number of $H$.
We choose a diagonal preconditioning matrix
$P$, with $P_{ii} = 1$ for
$1\le i\le a$ and $P_{ii}=p$ for $a<i\le b$,
where $a$ and $b$ are the dimensions
of $\alpha$ and $\beta$ respectively, and the
indices of $H$ corresponding to $\alpha$ precede
those corresponding to $\beta$. We take
$$p = \sqrt{\frac{\mathrm{median}(\mathrm{abs}(H_{1{:}a,
1{:}a}))}
{\mathrm{median}(\mathrm{abs}(H_{(a+1){:}(a+b), (a+1){:}(a+b)}))}}.$$
With this choice, the median absolute value of the upper left $a\times a$
block of $PHP$ is the same as that in its lower right $b\times b$ block.
If $PHP$ has a smaller condition number than $H$ does, then we use
the preconditioning, otherwise not.
\art{added note about preconditioning}


\section{Examples}\label{sec:examples}

Here we illustrate the use of convex optimization
to find an improved mixture importance sampler.
We use a two step process.  A preliminary sample
has $\alpha_j=1/J$ for all $j=1,\dots,J$.  On this
sample we estimate optimal $\alpha_j$ subject to
constraints $\alpha_j\ge \epsilon>0$.  Then we
use the second sample to estimate the desired expectation.

There are two main classes of problems where importance
sampling is crucial: spiky or even singular functions,
and rare events. In these examples we model the imperfect
information one might have in practice when choosing $q_j$. 
Therefore we include distributions $q_j$ that are good but not quite
optimal as well as some $q_j$ that correspond to poor choices.
The method does better if it can give a small $\alpha_j$
to the unhelpful $q_j$.

\subsection{Singular function}

For $\bsx \in \mathbb{R}^d$ we take the
nominal distribution to be $p(\bsx) = \dnorm(\mu,
  \Sigma)$ and $f(\bsx) = \|\bsx - \bsx_0\|^{-\gamma}$,
for some $\bsx_0$. We choose $\gamma < {d}/{2}$ so that
our example has $\var_p(f(\bsx)) < \infty$.
For this example, we use $d = 5$, $\gamma = 2.4$, $\Sigma_{i,j}
= \rho^{|i-j|}$ where $\rho = 1/2$ as well as
\[
\mu = \colvec{5}{0}{0}{0}{0}{0},\quad\text{and}\quad \bsx_{0} = \colvec{5}{1}{1}{1}{1}{1}.
\]

We choose a sequence of proposal densities
$q_j=\mathcal{N}(\bsx_k,  2^{-r}I_5)$, for $k = 1,
\dots, 5$, $r = 1, \dots, 10$, with index $j = 10(k-1) + r$,
where
\[
\bsx_1 = 
\bsx_0 =\colvec{5}{1}{1}{1}{1}{1}, 
\bsx_2 = \colvec{5}{-1}{-1}{-1}{-1}{-1}, 
\bsx_3 = \colvec{5}{-1}{1}{1}{1}{1}, 
\bsx_4 = \colvec{5}{1}{-1}{-1}{-1}{-1}, 
\bsx_5= \colvec{5}{-1}{-1}{1}{1}{1}.
\]
We include a defensive component via $q_{51} = p$. Note that we have
moved the defensive component from $j=1$ to $j=J$.
The only singularity is at $\bsx_0$.
We use a safety
lower bound $\epsilon = {0.1}/{51}$ on all $\alpha_j$.
That is, no proposal density will get less than $10$\% of
an equal share. This still allows the algorithm to allocate
the remaining $90$\% of the sample effort in a way to minimize variance.

For each simulation, we use $2$ batches of sizes $n_1 = 10^4$,
and $n_2 = 5\times10^5$ respectively. 
After sampling the first $n_1$ values, we optimize~\eqref{eq:criterion} over $\alpha$
with $\beta$ fixed at zero to get $\alpha_*$. We also optimize~\eqref{eq:criterion} over 
both $\alpha$ and $\beta$ getting $(\alpha_{**},\beta_{**})$.
These choices allow us to compare methods that use control variates with
those that do not.

To study mixture optimization with control variates we take
the second sample from $q_{\alpha_{**}}$.
Our estimate of $\mu$ is based on the second sample only:
\art{I took all the bolding off of $\beta$ as it is already
understood to be a vector.}
\begin{align*}
\hat{\mu}_{\alpha_{**},\beta} &= \frac{1}{n_2} 
  \sum\limits_{i = 1}^{n_2}\frac{f(x_{i})p(x_{i}) - {\beta}^{\tran}\bigl(\boldsymbol{q}(x_{i}) - p(x_{i})
\boldsymbol{1}\bigr)}{q_{\alpha_{**}}(x_{i})}.
\end{align*}
The value of $\beta$ used there is found by least squares fitted
to the second sample value, instead of using $\beta_{**}$.
The vector $\boldsymbol{q}(x_i)$ has elements $q_j(x_i)$
(ordinarily just $J-1$ of them)
and $\boldsymbol{1}$ has all $1$s.

To judge the effect of dropping control variates we also considered
taking the second sample from $q_*$ and using the estimate
\begin{align*}
  \hat{\mu}_{\alpha_{*}}&= \frac{1}{n_2} 
  \sum\limits_{i = 1}^{n_2}\frac{f(x_{i})p(x_{i})}{q_{\alpha_{*}}(x_{i})}.
\end{align*}
The judge the effect of non-uniform sampling,
we take $n_{2}$ samples from the uniform mixture $q_{U}$ which has
$\alpha_{j} = 1/J$ for $j = 1, \dots, J$.
Then we compute
\begin{align*}
\hat{\mu}_{q_U} &= \frac{1}{n_2} 
  \sum\limits_{i = 1}^{n_2}\frac{f(x_{i})p(x_{i})}{q_{U}(x_{i})},\quad\text{and}\\
\hat{\mu}_{q_U,\beta} &= \frac{1}{n_2} 
  \sum\limits_{i = 1}^{n_2}\frac{f(x_{i})p(x_{i}) -
    {\beta}^{\tran}(\boldsymbol{q}(x_{i}) -
    p(x_{i})\boldsymbol{1})}{q_{U}(x_{i})}
\end{align*}
Here ${\beta}$ is estimated with these $n_2$
samples from the uniform mixture.

We ran $5000$ replicates of these simulations.
Because $\mu$ is not known we estimate the mean
squared error for each method by its sample variance.
The Monte Carlo variance was found by running $N=5\times 10^7$ IID
samples.  
Table~\ref{tab:output_ex1_1} reports the results.
We see that for the singular function, optimizing the sampling weights
$\alpha$ reduces variance by about twenty-fold compared to using
the uniform distribution on weights.  Incorporating a control variate
brings a modest improvement.

\begin{table}[t]
\centering 
\begin{tabular}{cccrrrr}
  \toprule
Sampling & CV & $\hat\mu$ & VRFmc & VRF.uis \\ 
  \midrule
$U$ & No  & 0.173862 & 1.494 & 1.000 \\ 
$U$ & Yes  & 0.173871 & 3.273 & 2.190 \\ 
$\alpha_{*}$ & No  & 0.173867 & 28.321 & 18.953 \\ 
$\alpha_{**}$ & Yes  & 0.173867 & 33.013 & 22.093 \\ 
\bottomrule
\end{tabular}
\caption{Singular function example.
The estimate $\hat\mu$ was computed from $500{,}000$
observations using the sampler given in the first column. 
Control variates were used in two of those samples. The final 
columns give variance reduction factors compared to plain 
Monte Carlo and compared to uniform mixture importance 
sampling with no control variates. 
}
\label{tab:output_ex1_1}
\end{table}

We also present mixture components with the top 10 values of estimated
$\alpha$ in Table \ref{tab:alpha_ex1_opt.both_1} and \ref{tab:alpha_ex1_opt.alpha_1}.
About half of the sample is allocated to the defensive component  in both settings.
The true singularity point is at $k=1$. That point gets chosen more often
than the others, especially when no control variates are used. 
\begin{table}[t]
\centering
\begin{tabular}{cllcc}
\toprule
 k & $\sigma^2_r$ & Mean $\alpha_j$ & SD $\alpha_j$ \\ 
\midrule
D & D & 0.5263 & 0.028722 \\ 
1 & 0.25 & 0.2721 & 0.029184 \\ 
1 & 0.5 & 0.0912 & 0.037610 \\ 
2 & 0.5 & 0.0065 & 0.005128 \\ 
4 & 0.5 & 0.0039 & 0.002168 \\ 
1 & 0.0625 & 0.0031 & 0.001171 \\ 
1 & 0.125 & 0.0031 & 0.001210 \\ 
1 & 0.03125 & 0.0027 & 0.000511 \\ 
3 & 0.25 & 0.0026 & 0.000563 \\ 
3 & 0.5 & 0.0025 & 0.000591 \\ 
\bottomrule
\end{tabular}
\caption{Top 10 mixture components $\dnorm(\bsx_k,\sigma^2_rI_5)$
for the singular integrand in $\alpha_{**}$, which uses control variates. 
D denotes the defensive mixture. 
The last columns are mean and sd of $\alpha_j$ over $5000$ simulations.} 
\label{tab:alpha_ex1_opt.both_1}
\end{table}
\begin{table}[t]
\centering
\begin{tabular}{cllrr}
\toprule
k & $\sigma^2_r$ & Mean $\alpha_j$ & SD $\alpha_j$ \\ 
\midrule
D & D & 0.4838 & 0.018307 \\ 
1 & 0.25 & 0.2839 & 0.021976 \\ 
1 & 0.5 & 0.0971 & 0.019779 \\ 
1 & 0.0625 & 0.0290 & 0.006382 \\ 
1 & 0.03125 & 0.0099 & 0.002762 \\ 
1 & 0.125 & 0.0050 & 0.005936 \\ 
1 & 0.015625 & 0.0029 & 0.000733 \\ 
3 & 0.25 & 0.0022 & 0.000421 \\ 
3 & 0.5 & 0.0021 & 0.000358 \\ 
1 & 0.0078125 & 0.0021 & 0.000353 \\ 
\bottomrule
\end{tabular}
\caption{Top 10 mixture components $\dnorm(\bsx_k,\sigma^2_rI_5)$
for the singular integrand in $\alpha_{*}$, which does not use control variates. 
D denotes the defensive mixture. 
The last columns are mean and sd of $\alpha_j$ over $5000$ simulations.} 
\label{tab:alpha_ex1_opt.alpha_1}
\end{table}

\subsection{Rare event}

For our rare event simulation we took
$\bsx \in \mathbb{R}^3$,  $p(\bsx) = \mathcal{N}(\mu,  I_3)$ 
and defined event sets
$$\mathcal{D}_i = \{\bsx\mid\Vert\bsx\Vert >
\eta_i\Vert\bsz_{i}\Vert \text{ elementwise and sign$(\bsx)$ $=$ sign$(\bsz_{i})$} \}$$
for $i = 1, \dots, 8$. 
The points $\bsz_i$ are all eight points 
of the form $(\pm \eta_i,\pm\eta_i,\pm\eta_i)$ for $\eta_i>0$, with $\eta_i$
chosen so that $\Pr(\bsx\in D_i) = 16^{-i}\times10^{-3}$.
\art{I now think we might not need to identify them
since index i does it.  But a referee might ask later.}

The sets $\cd_i$ are disjoint.  
The integrand $f$ is $1$ if $\bsx\in \cd = \cup_{i=1}^8\cd_i$ 
and is $0$ otherwise. 
For this problem we know that 
$$\mu =
\e(f(\bsx)) = \sum_{i = 1}^{8} \Pr(\bsx\in\mathcal{D}_i)
\doteq 6.666666665\times10^{-5}$$
and in a Monte Carlo sample of $n$ observations the variance
of the sample mean is $\mu(1-\mu)/n\doteq\mu/n$.

We choose $q_j=\dnorm(\bsz_k,
  \sigma_r^2{I}_{3})$, for $k = 0, \dots, 8$ as the mixture
components, including a new point
$\bsz_0 = \begin{pmatrix}0 & 0 & 0\end{pmatrix}^\tran$
and the previous corners $\bsz_k, k = 1, \cdots, 8$.
That is, we suppose that the user has some idea where the
rare events are.  But we suppose that there is little knowledge
of the right sampling variance, so we take
\art{fixed typo}
$$\sigma^2_r \in\Bigl\{ \frac1{50}, \frac1{40}, \frac1{30}, 
\frac1{20}, \frac1{10}, \frac12, 2, 10, 20, 30, 40, 50\Bigr\}$$ 
for $r = 1, \dots, 12$ respectively. In addition to these $108$ mixture
components, we include the defensive
component as $q_{109} = p$, once again as $q_J$, not $q_1$.

We take a pilot sample of $n_1=10^4$ observations to estimate mixture
components with, and then we  follow up with $n_2 = 100{,}000$
observations for the final estimate.
We optimized with safety bound $\epsilon = {0.1}/{109}$ for all
$\alpha$. This simulation was
repeated $5000$ times.

\begin{table}[ht]
\centering
\begin{tabular}{cccrrrrrrr}
\toprule
Sampling & CV & $10^5\hat\mu$ & VRFmc & VRFuis & $\overline{\text{MSE}/\wh{\text{VAR}}}$\\
\midrule
U& No  & 6.669473 & 65.249 & 1.000 &  0.980 \\ 
U & Yes  & 6.668018 & 69.516 & 1.065 &  0.976 \\ 
$\alpha_*$ & No & 6.667731 & 1663.983 & 25.502 & 1.018 \\ 
$\alpha_{**}$ & Yes  & 6.667475 & 1676.611 & 25.696  & 1.015 \\ 
\bottomrule
\end{tabular}
\caption{Rare event example.
The estimate $\hat\mu$ was computed from $100{,}000$
observations using the sampler given in the first column. 
Control variates were used in two of those samples. The next
columns give variance reduction factors compared to plain 
Monte Carlo and compared to uniform mixture importance 
sampling with no control variates.  The final column compares
actual squared error with its sample estimate.
} 
\label{tab:output_ex2_1}
\end{table}

The results are in Table~\ref{tab:output_ex2_1}.
For this problem we can use the actual Monte Carlo variance
when computing the variance reduction factor with respect to
plain Monte Carlo.  The method of control variates also provides
at estimate $\wh\var(\hat\mu_{\alpha,\beta})$ which comes from
the regression model variance estimate of the intercept term.
The last column in Table~\ref{tab:output_ex2_1} shows the
average over $5000$ simulations of the true mean squared error
divided by its estimate.
Here we see about a $25$-fold gain from optimizing the mixture
weights.  There is almost no gain from using the control variate.

The top $10$ mixture components in the rare event 
problem are listed in 
Table~\ref{tab:alpha_ex2_opt.alpha_1}. That table has values for the mixture
sampling without control variates. Including control variates gave the same ten
components in the same order with nearly identical mixing probabilities.
The method samples very little from the mixture components with large
variances. The largest weights go on the least rare component corners
$i=1,2,3$.


\begin{table}[t]
\centering
\begin{tabular}{rllcc}
\toprule
k & $\sigma^2$ & Mean $\alpha_j$ & SD $\alpha_j$ \\ 
\midrule
1 & 0.5 & 0.5188 & 0.064185 \\ 
1 & 0.1 & 0.3196 & 0.071853 \\ 
2 & 0.5 & 0.0266 & 0.006733 \\ 
2 & 0.1 & 0.0249 & 0.005421 \\ 
1 & 0.05 & 0.0058 & 0.019458 \\ 
1 & 0.02 & 0.0028 & 0.003635 \\ 
1 & 0.033 & 0.0020 & 0.004237 \\ 
1 & 0.025 & 0.0018 & 0.002479 \\ 
3 & 0.5 & 0.0014 & 0.000338 \\ 
2 & 0.05 & 0.0013 & 0.000706 \\ 
\bottomrule
\end{tabular}
\caption{Top 10 mixture components $\dnorm(\bsz_k,\sigma^2_rI_2)$
for the singular integrand in $\alpha_{*}$, which does not use control variates. 
D denotes the defensive mixture. 
The last columns are mean and sd of $\alpha_j$ over $5000$ simulations.} 
\label{tab:alpha_ex2_opt.alpha_1}
\end{table}

\subsection{Summary}
In both of our examples there was a gain from optimizing over $\alpha$
but very little gain from employing the control variate.
It can happen that a control variate makes an enormous
improvement. One such case is example 2 in \cite{owen:zhou:2000}.


Table \ref{tab:running_time}
gives the average running time for the four methods we compare on
the two examples. These times include running the pilot sample,
any optimizations needed, and then running the final sample.
\art{fixed typo}
The more complicated methods take longer but not
enough to outweight the gain from optimizing $\alpha$.
In three of the examples, the gain from control variates  becomes
a small loss when computation time is accounted for. The exception
is that for the singular function with uniform mixture sampling
the control variate roughly halves the variance while taking about $37$\%
longer, resulting in a minor efficiency gain.

\begin{table}[t]
\centering
\begin{tabular}{rllrrrr}
\toprule
  & $n_1$ & $n_2$ & $U$ no CV & $U$ CV &$\alpha_*$ no CV &  $\alpha_{**}$ CV \\
\midrule
Singularity & $10^{4}$ & $5\times 10^5$  &13.49 & 18.44 & 16.25 &23.82\\
Rare event& $10^{4}$ & $10^{5}$ & 4.15 & 6.19 & 9.61 &23.28\\
\bottomrule
\end{tabular}
\caption{
Average running times in seconds for four estimators on two examples.
}
\label{tab:running_time}
\end{table}

\section{Discussion}\label{sec:conclusions}

We have developed a second version of mixture
importance sampling with control variates
$\hat\mu_{\alpha,\beta}$ at~\eqref{eq:hatmuab}.
For defensive mixtures, this estimator
has the same minimal variance as
$\tilde\mu_{\alpha,\beta}$ at~\eqref{eq:tildemuab}
presented in~\cite{owen:zhou:2000}.
In their stratified versions these estimators enjoy both the regret bounds of the
balance heuristic as well as the regret bounds of
mixture importance sampling.
The new estimator has a variance that is jointly convex
in $\alpha$ and $\beta$.  As a result, we can optimize
a sample mean square over both
of these vectors and then use the resulting $\alpha$
for a followup sample.

There is a large literature on adaptive importance sampling,
going back at least to \cite{mars:1956}.
The cross-entropy method \citep{rubi:kroe:2004}
is devoted to finding a single importance distribution
$q(\bsx) = q(\bsx;\theta)$ for a parameter $\theta\in\Theta$.
It can be challenging to optimize the variance over $\theta$
and they optimize instead a (sample) entropy criterion comparing
$q$ to the optimal density proportional to $fp$ (for $f\ge0$).
\cite{capp:douc:guil:mari:robe:2008} seek to optimize a mixture
$\sum_j\alpha_jq(\bsx,\theta_j)$ over both $\alpha_j$ and $\theta_j$.
They use an entropy criterion instead of variance and apply an
EM style algorithm to estimate the optima.  They note that
the entropy is convex in the mixing parameters, but that it is not
jointly convex in those parameters and the $\theta_j$.
Our methods may be used in conjunction with theirs. Given a
fixed set of $\theta_j$ selected by some other method, we can
jointly estimate the variance-optimal mixing parameters $\alpha_j$
and control variate coefficients $\beta_j$ via convex optimization
applied to sample values.

A multiple of $\mu^2/n$ appears in the regret bound
for multiple importance sampling versus the unknown
best mixture component.  Using a control variate
removes that term.
The quantity $\mu^2/n$ is often negligible
compared to the attained sample variance.
For instance in plain Monte Carlo sampling of
a rare event, the variance $\mu(1-\mu)/n$
is far larger than $\mu^2/n$ and the latter
may be ignored.  But in some rare event simulations
it is possible to obtain estimator with bounded
relative error \citep{heid:1995}. That means that for a sequence of problems with
$\mu\downarrow0$ the variance  is $O(\mu^2/n)$. In
problems with such very efficient estimators, removing
the multiple of $\mu^2/n$ from the variance bound,
as density-based control variates do, is a meaningful improvement.

\section*{Acknowlegment}
This work was supported by grants DMS-0906056
and DMS-1407397
from the National Science Foundation.
\bibliographystyle{apalike}
\bibliography{mcbook}

\section{Proofs}

Theorem~\ref{thm:smallregret}
was proved earlier by~\cite{owen:zhou:2000}.
Their notation defined control variates via
integrals, not expectations.
We give a direct proof in the notation of the
present paper to make this article
self-contained.

\subsection{Proof of Theorem~\ref{thm:smallregret}}
\label{sec:thm:smallregret}

Choose $k\in\{1,2,\dots,J\}$ and then
select $\beta$ with $\beta_k = 0$
and $\beta_j = -\mu {\alpha_j}/{\alpha_k}$ for $j\ne k$, giving
$\sum\limits_{\ell = 1}^{J}\beta_\ell = \mu(1 - {1}/{\alpha_k})$.
From~\eqref{eq:variscvtilde} we find that
\begin{align*}
n\var(\wt\mu_{\alpha,\beta})
& =
\int \frac{[fp-\mu q_\alpha-\beta^\tran hp
+\beta^\tran\one q_\alpha]^2}{q_\alpha}\rd\bsx
\end{align*}
The condition on $q_j$ yields $h_jp=q_j$,
and then for this $\beta$ we find that
$n\var(\wt\mu_{\alpha,\beta})$ is
\begin{align*}
&
\int \frac{\Bigl[fp-\Bigl(\mu\alpha_k-\sum\limits_{\ell = 1}^{J}\beta_\ell\alpha_k +\beta_k\Bigr)q_k -\sum_{j\ne k}
\Bigl(\mu\alpha_j-\sum\limits_{\ell = 1}^{J}\beta_\ell\alpha_j+\beta_j\Bigr)q_j\Bigr]^2}{\sum_{j=1}^J\alpha_jq_j}\rd\bsx\\
&=
\int \frac{(fp-\mu q_k)^2}{\sum_{j=1}^J\alpha_jq_j}\rd\bsx
\le \int \frac{(fp-\mu q_k)^2}{\alpha_kq_k}\rd\bsx = \frac{\sigma^2_{q_k}}{\alpha_k}.
\end{align*}
The result follows because $\var(\wt\mu_{\alpha,\beta_*})
\le \var(\wt\mu_{\alpha,\beta})$.
$\Box$

\subsection{Proof of Theorem~\ref{thm:smallregrethat}}
\label{sec:thm:smallregrethat}

The proof is similar to that of Theorem~\ref{thm:smallregret}. 
We choose $k\in\{1,2,\dots,J\}$ and define $\beta_j = \mu(1_{j=k}-\alpha_k)$.
Making this substitution leads to
$$
\hat\mu_{\alpha,\beta} = \mu + \frac1n\sum_{i=1}^n\frac{f(\bsx_i)p(\bsx_i)
-\mu q_k(\bsx_i)}{q_\alpha(\bsx_i)},
$$
which has variance
$$
\frac1n\int (fp-\mu q_k)^2/q_\alpha\rd\bsx
\le\frac1{n\alpha_k}\int (fp-\mu q_k)^2/q_k\rd\bsx=\sigma^2_{q_k}/n\alpha_k.
\quad\Box
$$

\subsection{Proof of Corollaries~\ref{cor:smallregretbeta}
and~\ref{cor:smallregrethatbeta}
}\label{sec:cor:smallregretbeta}

The same argument underlies both proofs.
We give the version for Corollary~\ref{cor:smallregretbeta}.

Choose $j\in\{1,\dots,J\}$ and $\beta_j\in\real^J$
to minimize $\sigma^2_{q_j,\beta_j}$ and
define $g(\bsx) = f(\bsx)-\beta_j^\tran h(\bsx)$.
Now we apply Theorem~\ref{thm:smallregret}
to the integrand $g$ and let $\beta_*$
be the minimizing regression coefficient
from that Theorem.
Writing $g(\bsx)-\beta_*^\tran h(\bsx) =
f(\bsx) - (\beta_*-\beta_j)^\tran h(\bsx)$
we see that $\beta = \beta_*-\beta_j$ satisfies
the conclusion of this Corollary.
$\Box$

\subsection{Proof of Lemma~\ref{lem:samespace}}
\label{sec:lem:samespace}

Under the conditions of the Lemma, $\theta=\one$.
Let $\gamma_j=\beta_j-\alpha_j\beta^\tran\one$.
Now for any $\bsx$ sampled from $q_\alpha$,
we necessarily have $q_\alpha(\bsx)>0$. Then
\begin{align*}
\gamma^\tran(h(\bsx)-\one)p(\bsx)/q_\alpha(\bsx)
&=\gamma^\tran(q(\bsx)-p(\bsx))/q_\alpha(\bsx)\\
& = \sum_{j=1}^J
(\beta_j-\alpha_j\beta^\tran\one)(q_j(\bsx)-p(\bsx))/q_\alpha(\bsx)\\
& = \sum_{j=1}^J\beta_jq_j(\bsx)/q_\alpha(\bsx)
-\beta^\tran\one p(\bsx)/q_\alpha(\bsx)
-\beta^\tran\one
+\beta^\tran\one p(\bsx)/q_\alpha(\bsx)\\
& = \sum_{j=1}^J\beta_jq_j(\bsx)/q_\alpha(\bsx)
-\beta^\tran\theta,
\end{align*}
after cancelling and using $\theta=\one$.
Substituting this last expression into
each term for $\hat\mu_{\alpha,\gamma}$
for~\eqref{eq:hatmuab} yields $\tilde\mu_{\alpha,\beta}$
in~\eqref{eq:tildemuab}.
$\Box$

\subsection{Proof of Theorem~\ref{thm:convexalpha}}
\label{sec:thm:convexalpha}
First we write
\begin{align}\label{eq:varsigsq2c}
\sigma^2_\alpha + \mu^2
=
\int_D \frac{\bigl( f(\bsx)p(\bsx)\bigr)^2}{\sum_{j=1}^J\alpha_jq_j(\bsx) }\rd \bsx
\end{align}
where $D = \{\bsx\mid q_\alpha(\bsx)>0\}$ is the common domain of
$q_\alpha$ for all $\alpha\in S^0$.
For fixed $\bsx\in D$ and $r,s\in\{1,\dots,J\}$,
$$
\frac{\partial^2}{\partial \alpha_r\partial\alpha_s}
\frac{ f(\bsx)^2p(\bsx)^2}{
\sum_{j=1}^J\alpha_jq_j(\bsx) }
=\frac{2f(\bsx)^2p(\bsx)^2q_r(\bsx)q_s(\bsx)}{
\Bigl(\sum_{j=1}^J\alpha_jq_j(\bsx) \Bigr)^3}
$$
and so the integrand in~\eqref{eq:varsigsq2c} has a
positive semi-definite Hessian matrix at
every $\bsx$.  Thus the integrand is convex in $\alpha$
for each $\bsx$, and hence $\sigma^2_\alpha$ is a convex
function of $\alpha$. $\Box$

\subsection{Proof of Theorem~\ref{thm:jointconvex}}
\label{sec:thm:jointconvex}

Similar to Theorem~\ref{thm:convexalpha}, we write
\begin{align}\label{eq:varsigsq2b}
\sigma^2_{\alpha,\beta} + \mu^2
=
\int_{D(q_\alpha)} \frac{\bigl( f(\bsx)p(\bsx) - \sum_{J=1}^J\beta_j(q_j(\bsx)-p(\bsx))
\bigr)^2}{\sum_{j=1}^J\alpha_jq_j(\bsx) }\rd \bsx.
\end{align}

The function $y^2/z$ is jointly convex in $z\in(0,\infty)$ and $y\in \real$.
For fixed $\bsx\in D$ the integrand in~\eqref{eq:varsigsq2b} has
the form $y^2/z$ where $y$ is an affine function of $\beta$
and $z$ is a positive linear function of $\alpha\in S^0$.
As a result the integrand
in~\eqref{eq:varsigsq2b} is jointly convex in $(\alpha,\beta)$
and so also is $\sigma^2_{\alpha,\beta}$. $\Box$

\subsection{Proof of Theorem~\ref{thm:safetyepsilon}}
\label{sec:thm:safetyepsilon}

We use Lemma~\ref{lem:mixvarbound} in our proof.
It is of independent interest.

\begin{lemma}\label{lem:mixvarbound}
Let $\alpha,\gamma,\omega\in S$ with
$\gamma = \lambda \alpha + (1-\lambda)\omega$
and $0<\lambda\le1$.
Then
\begin{align*}
&\min_\beta\var( \tilde\mu_{\gamma,\beta} )\le \frac1{n\lambda}\sigma^2_\alpha,
\quad\text{and}\\
&\min_\beta\var( \hat\mu_{\gamma,\beta} )\le \frac1{n\lambda}\sigma^2_\alpha.
\end{align*}
\end{lemma}
\begin{proof}
For $J=2$, let $\wt q_1=q_\alpha$
and $\wt q_2 = q_\omega$.
Then $q_\gamma = \wt \alpha_1 \wt q_1 +\wt\alpha_2\wt q_2$
for $\wt\alpha = (\lambda,1-\lambda)$.
The first part then follows from
Theorem~\ref{thm:smallregret} using $\wt\alpha$
and $\wt q_j$ in place of $\alpha$ and $q_j$.
The second part follows from
Theorem~\ref{thm:smallregrethat}, noting that
for a nominal $\wt q_1=q_\alpha$ the mixture
weights $\wt\alpha$  include a defensive component.
\end{proof}

Now we prove Theorem~\ref{thm:safetyepsilon} itself.
First we show that for any $\alpha\in S\setminus S^\epsilon$
we can find $\delta\in S^\epsilon$ with
$\var(\tilde\mu_{\delta,0})\le (1+\epsilon)\var(\tilde\mu_{\alpha,0})$.
Recall that $K=\#\{1\le j\le J\mid \alpha_j<\epsilon\}$.
We have $1\le K\le J-1$; the lower limit comes from
$\alpha\not\in S^\epsilon$ and the upper limit from
$\epsilon<1/d$.

Now define $\omega\in S^\epsilon$
with $\omega_j = \epsilon$ for $\alpha_j\ge\epsilon$
and $\omega_j = \Omega \equiv (1-\epsilon(J-K))/K$
for $\alpha_j<\epsilon$.
Define $\lfloor \alpha\rfloor \equiv \min_j\alpha_j$
and then set $\lambda = (\Omega-\epsilon)/(\Omega-\lfloor\alpha\rfloor)$
and $\delta_j = \lambda\alpha_j+(1-\lambda)\omega_j$.
By construction $0<\lambda< 1$. If $\alpha_j\ge\epsilon$
then $\delta_j \ge\min(\alpha_j,\omega_j)=\epsilon$.
Otherwise
\begin{align*}
\delta_j \ge
\frac{\Omega-\epsilon}{\Omega-\lfloor\alpha\rfloor}
\lfloor\alpha\rfloor +
\frac{\epsilon-\lfloor\alpha\rfloor}{\Omega-\lfloor\alpha\rfloor}\Omega=\epsilon.
\end{align*}
It follows that $\delta\in S^\epsilon$. Now from Lemma~\ref{lem:mixvarbound}
we find
\begin{align*}
\var(\tilde\mu_{\delta,0})&
\le \frac1\lambda\var(\tilde\mu_{\alpha,0})
\end{align*}
and
$$
\frac1\lambda =
\frac{\Omega-\lfloor\alpha\rfloor}{\Omega-\epsilon}
\le\frac{\Omega}{\Omega-\epsilon}
=\frac{1-\epsilon(J-K)}{1-\epsilon(J-K)-K\epsilon}
=\frac{1-\epsilon J+\epsilon K}{1-\epsilon J}.
$$

We can extend this result to the
$\var(\tilde\mu_{\alpha,\beta})$ by the same
device used in Corollary~\ref{cor:smallregretbeta}.
Then because $\alpha_1>0$ when $\alpha\in S^\epsilon$,
a defensive component is present, and the
results then apply to $\var(\hat\mu_{\alpha,\beta})$.
$\Box$

\end{document}

%% file: newcommandsnotbook.tex
\allowdisplaybreaks

\newcommand{\cd}{{\cal D}}

\newcommand{\bsx}{\boldsymbol{x}}

\newcommand{\bsz}{\boldsymbol{z}}

\newcommand{\real}{\mathbb{R}}

\newcommand{\tran}{\mathsf{T}} 

\newcommand{\rd}{\mathrm{\, d}}

\newcommand{\one}{{\mathbf{1}}}

\newcommand{\var}{{\mathrm{Var}}}
\newcommand{\cov}{{\mathrm{Cov}}}

\newcommand{\wh}{\widehat}
\newcommand{\wt}{\widetilde}

\renewcommand{\ge}{\geqslant}
\renewcommand{\le}{\leqslant}

\newcommand{\dnorm}{\mathcal{N}}

\newcommand{\e}{{\mathbb{E}}} 

\newcommand{\strat}{{\mathrm{strat}}}



%% file: newcommandsarticle.tex
\newcount\colveccount
\newcommand*\colvec[1]{
        \global\colveccount#1
        \begin{pmatrix}
        \colvecnext
}
\def\colvecnext#1{
        #1
        \global\advance\colveccount-1
        \ifnum\colveccount>0
                \\
                \expandafter\colvecnext
        \else
                \end{pmatrix}
        \fi
}